\documentclass[onefignum,onetabnum]{siamonline190516}

\usepackage{braket,amsfonts}

\usepackage{array}

\usepackage[caption=false]{subfig}

\newsiamthm{claim}{Claim}
\newsiamremark{remark}{Remark}
\newsiamremark{hypothesis}{Hypothesis}
\crefname{hypothesis}{Hypothesis}{Hypotheses}

\usepackage{algorithmic}

\usepackage{graphicx,epstopdf}

\crefname{ALC@unique}{Line}{Lines}

\usepackage{amsopn}

\newcommand{\ddd}{\scriptsize{\textnormal{d}}}
\newcommand{\dd}{\textnormal{d}}

\newcommand*{\Glp}[1][$n$]{\text{GL}^+(#1)}

\newcommand{\SO}{\textnormal{SO}}
\newcommand{\E}{\textnormal{E}}
\newcommand{\SE}{\textnormal{SE}}
\newcommand{\Exp}{\textnormal{Exp}}
\newcommand{\Log}{\textnormal{Log}}
\newcommand{\Ad}{\textnormal{Ad}}
\newcommand{\inv}{\textnormal{inv}}

\usepackage[utf8]{inputenc}
\usepackage{graphicx}
\usepackage{xcolor}
\usepackage{overpic}
\usepackage{wrapfig,lipsum}

\usepackage[colorinlistoftodos]{todonotes}

\usepackage{amssymb, amsfonts, amsopn, mathrsfs}

\newcommand\blfootnote[1]{
    \noindent
    \begingroup
    \renewcommand\thefootnote{}\footnote{#1}%
    \addtocounter{footnote}{-1}%
    \endgroup
}

\usepackage{array}

\DeclareRobustCommand{\ShowColormap}{\raisebox{-0.14em}{\includegraphics[height=.8em]{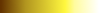}}}

\usepackage{xspace}
\usepackage{bold-extra}
\usepackage[most]{tcolorbox}

\colorlet{texcscolor}{blue!50!black}
\colorlet{texemcolor}{red!70!black}
\colorlet{texpreamble}{red!70!black}
\colorlet{codebackground}{black!25!white!25}

\lstdefinestyle{siamlatex}{%
  style=tcblatex,
  texcsstyle=*\color{texcscolor},
  texcsstyle=[2]\color{texemcolor},
  keywordstyle=[2]\color{texemcolor},
  moretexcs={cref,Cref,maketitle,mathcal,text,headers,email,url},
}

\tcbset{%
  colframe=black!75!white!75,
  coltitle=white,
  colback=codebackground, 
  colbacklower=white, 
  fonttitle=\bfseries,
  arc=0pt,outer arc=0pt,
  top=1pt,bottom=1pt,left=1mm,right=1mm,middle=1mm,boxsep=1mm,
  leftrule=0.3mm,rightrule=0.3mm,toprule=0.3mm,bottomrule=0.3mm,
  listing options={style=siamlatex}
}

\newtcblisting[use counter=example]{example}[2][]{%
  title={Example~\thetcbcounter: #2},#1}

\newtcbinputlisting[use counter=example]{\examplefile}[3][]{%
  title={Example~\thetcbcounter: #2},listing file={#3},#1}

\DeclareTotalTCBox{\code}{ v O{} }
{ 
  fontupper=\ttfamily\color{black},
  nobeforeafter,
  tcbox raise base,
  colback=codebackground,colframe=white,
  top=0pt,bottom=0pt,left=0mm,right=0mm,
  leftrule=0pt,rightrule=0pt,toprule=0mm,bottomrule=0mm,
  boxsep=0.5mm,
  #2}{#1}

\patchcmd\newpage{\vfil}{}{}{}
\begin{tcbverbatimwrite}{tmp_\jobname_header.tex}
\title{Bi-invariant Dissimilarity Measures for Sample Distributions in Lie Groups\thanks{This work is an extension of the conference paper~\cite{HanikHegevonTycowicz2020}. Submitted to the editors April 6, 2021.
\funding{M. Hanik is funded by the Deutsche Forschungsgemeinschaft (DFG, German Research Foundation) under Germany’s Excellence Strategy – The Berlin Mathematics Research Center MATH+ (EXC-2046/1, project ID: 390685689). This work was supported by the Bundesministerium fuer Bildung und Forschung (BMBF) through BIFOLD - The Berlin Institute for the Foundations of Learning and Data (ref. 01IS18025A and ref 01IS18037A).\\[-2ex]}}}

\author{Martin Hanik\thanks{Zuse Institute Berlin, Takustr.\ 7, 14195 Berlin, Germany (\email{hanik@zib.de}, \email{hege@zib.de}).}
\and Hans-Christian Hege\footnotemark[2]
\and Christoph von Tycowicz\thanks{Freie Universität Berlin, Kaiserswerther Str.\ 16-18, 14195 Berlin, Germany (\email{vontycowicz@zib.de}).} }

\headers{Bi-invariant Dissimilarity Measures for Sample Distributions in Lie Groups}{M. Hanik, H.-C. Hege, C. von Tycowicz}
\end{tcbverbatimwrite}
\input{tmp_\jobname_header.tex}

\ifpdf
\hypersetup{ pdftitle={Bi-invariant Dissimilarity Measures for Sample Distributions in Lie Groups} }
\fi

\begin{document}
\maketitle

\begin{tcbverbatimwrite}{tmp_\jobname_abstract.tex}
\begin{abstract}

    Data sets sampled in Lie groups are widespread, and as with multivariate data, it is important for many applications to assess the differences between the sets in terms of their distributions.
    Indices for this task are usually derived by considering the Lie group as a \emph{Riemannian} manifold.
    Then, however, compatibility with the group operation is guaranteed only if a bi-invariant metric exists, which is not the case for most non-compact and non-commutative groups.
    We show here that if one considers an affine connection structure instead, one obtains bi-invariant generalizations of well-known dissimilarity measures:
    a \emph{Hotelling $T^2$ statistic}, \emph{Bhattacharyya distance} and \emph{Hellinger distance}.
    Each of the dissimilarity measures matches its multivariate counterpart for Euclidean data and is translation-invariant, so that biases, e.g., through an arbitrary choice of reference, are avoided.
    We further derive non-parametric two-sample tests that are bi-invariant and consistent.
    We demonstrate the potential of these dissimilarity measures by performing group tests on data of knee configurations and epidemiological shape data. Significant differences are revealed in both cases.
\end{abstract}

\begin{keywords}
	non-metric statistics, two-sample test, shape analysis, Lie group
	\blfootnote{{\\ Data used in this article were obtained from the Alzheimer’s Disease Neuroimaging Initiative (ADNI) database (adni.loni.usc.edu). Investigators within the ADNI contributed to the design and implementation of the ADNI and/or provided data, but were not involved in the analysis or writing of this report. For a complete list of ADNI investigators see \url{http://adni.loni.usc.edu/wp-content/uploads/how_to_apply/ADNI_Acknowledgement_List.pdf}.}}
\end{keywords}

\end{tcbverbatimwrite}
\input{tmp_\jobname_abstract.tex}

\section{Introduction}
        Manifold-valued data occur in many applications \cite{Chirikjian2012,pennec2006intrinsic, peyre2009manifold,rahman2005multiscale} and it can often be conceived as elements of a Lie group. Examples include representations of skeletal systems (e.g., in robotics~\cite{ParkBobrowPloen1995}), inferences of anatomical structures by evaluating diffusion tensor fields in medicine~\cite{Arsigny_ea2006,arsigny2007geometric,bacak2016second,pennec2019riemannian}, position- and motion-independent recognition of objects in computer vision~\cite{huang2017deep,veeriah2015differential,vemulapalli2014human,vemulapalli2016rolling}, or analysis of covariance matrices (e.g., in feature-based image analysis~\cite{tuzel2006region}). The algorithmic tasks involved in Lie group-based computations are correspondingly diverse: They range from computation of geometric means \cite{moakher2002means, moakher2005differential}, to function approximation and regression \cite{Hanik_ea2020,samir2012gradient,wallner2007smoothness}, numerical solution of differential equations \cite{iserles2005lie}, numerical minimization \cite{taylor1994minimization}, signal processing \cite{barbaresco2008innovative,Barbaresco2020,Cesic_ea2016,fiori2009algorithm}, image analysis~\cite{bekkers2018roto,duits2007scale}, as well as computer vision~\cite{calabi1998differential,olver1994differential,pai2020geometric, porikli2006covariance,tosato2010multi,trouve1998diffeomorphisms}.
    
    Another area where Lie-group-valued measurements are regularly performed is shape analysis: The idea is to represent shapes of objects as deformations of a common reference, as introduced by D'Arcy Thompson over 100 years ago~\cite{Thompson1992}; the deformations are usually elements of a Lie group.
    For example, in~\cite{Adler_ea2002,Boisvert_ea2008}, configurations of the human spine are encoded in a product group consisting of translations and rotations. Classical matrix groups are also used in physically motivated shape spaces~\cite{AmbellanZachowTycowicz2021,vonTycowicz_ea2018} as well as in the characterization of volume~\cite{Woods2003} and surface~\cite{AmbellanZachowvonTycowicz2019_GL3} deformations.
    Important algorithmic tasks are matching, analysis and statistics of shapes \cite{fletcher2003statistics,huckemann2010intrinsic,srivastava2016functional, tabia2014covariance,trouve2005metamorphoses}.
    
    For data analysis in Lie groups, it is necessary to generalize methods of multivariate statistics to the Lie group setting.
    Keeping in mind that Lie groups possess symmetries since they act on themselves via translations, it is desirable to look for generalizations that respect these symmetries. 
    The task, then, is to derive methods that are invariant/equivariant under in-group translations (from both left and right).
    This property is not only an established theoretical criterion for selection of statistical methods~\cite{HallinJurevkova2014}, but it also has practical advantages: 
    In shape analysis, for example, it avoids a bias due to the choice of reference and, at the same time, different data configurations~\cite{Adler_ea2002,AmbellanZachowvonTycowicz2019_GL3,Woods2003}. In machine learning applications, equivariant convolutional networks have been shown to perform very well on Lie groups~\cite{CohenTacoWelling2016,Finzi_ea2020}. 
    
    Although geometrically-defined statistical methods from Riemannian geometry have long been used, they respect the symmetries only if a bi-invariant Riemannian metric exists on the Lie group at hand. However, this is often not the case: particularly relevant examples for applications are the group of rigid body transformations and the general linear group in dimensions greater than one. 
    
    To overcome the problems of the Riemannian approach, Pennec and Arsigny generalized the notions of mean, covariance and Mahalanobis distance to Lie groups using a (non-metric) affine structure~\cite{PennecArsigny2013}. The quantities defined in this way have the desired invariance/equivariance properties under translations. We build on 
    this work to derive bi-invariant (i.e., invariant under translation from both left and right) generalizations of the Hotelling $T^2$ statistic and Bhattacharyya distance for data in Lie groups that follow a generalized normal distribution.
    Both are well-known indices from multivariate statistics and are used to measure dissimilarities between two probability distributions. As such they have found numerous applications, for example, in statistical hypothesis testing~\cite{Pardo2005}, feature extraction~\cite{ChoiLee2003}, and image processing~\cite{GoudailRefregierDelyon2004}. In contrast to generalizations based on Riemannian structures~\cite{Hong2015,MuralidharanFletcher2012}, the proposed quantities are not only bi-invariant, but also reduce to the multivariate formulations in the case of Euclidean spaces.
    To illustrate the potential, we use the newly defined notions to construct hypothesis tests for two prospective applications. First, we detect differences in the configuration of human knees under osteoarthritis when compared to healthy controls. 
    Second, we test for differences in shape of the right hippocampus of patients in early Alzheimer's disease; thereby, both a local and a global test confirm effects known from the literature. 
    Python implementations of our proposed methods, which were also used for the tests, are available online as part of the Morphomatics library~\cite{Morphomatics}.
    
    This article is based in part on the workshop paper~\cite{HanikHegevonTycowicz2020}, but includes both theoretical and experimental extensions: We provide a rigorous mathematical formulation of the underlying concepts, including new discussions of the symmetry properties of the proposed indices under an exchange of the data sets and on the behavior of the indices under inversion of the data. Furthermore, we derive the connection between the bi-invariant Bhattacharyya distance and its version for densities and propose a bi-invariant Hellinger distance. Finally, we extend the previous experiments to include a global hypothesis test for the equality of shape distributions and a local test based on the Bhattacharyya distance and perform a new hypothesis test for differences in the SE(3)-valued configuration of the knee under osteoarthritis.

\section{Theoretical Background}
    Let $(p_1,\dots,p_m)$ and $(q_1,\dots,q_n )$ be two data sets in $\mathbb{R}^d$, each with i.i.d.\ normally distributed elements, and let $\overline{p}$ and $\overline{q}$ denote the respective sample means. Assuming homoscedasticity (i.e., that both distributions share the same covariance matrix) the data's pooled sample covariance is given by
    \begin{equation} \label{eq:Euclidean_pooled_covariance}
        \widehat{S} = \frac{1}{m + n - 2} \left( \sum_{i=1}^m (p_i - \overline{p}) (p_i - \overline{p})^T + \sum_{j=1}^n (q_j - \overline{q}) (q_j - \overline{q})^T \right).
    \end{equation}
    Since the covariances are the same, $\widehat{S}$ is an unbiased estimator of the pooled covariance.
    The multivariate Hotelling $T^2$ statistic is then defined as the square of the Mahalanobis distance scaled by $mn / (m + n)$:
    \begin{equation} \label{eq:Euclidean_test}
        t^2 \big((p_i),(q_i) \big) = \frac{mn}{m + n} (\overline{p} - \overline{q})^T \widehat{S}^{-1} (\overline{p} - \overline{q}).
    \end{equation}
    It measures the difference of $\overline{p}$ and $ \overline{q}$ weighted by the inverse of the pooled covariance. Therefore, directions in which high variability was observed are weighted less than those with little spreading around the corresponding components of the mean. 
    
    If additionally the variances of both distributions differ, the Bhattacharyya distance~\cite{Kailath1967} 
    is a suggested index for assessing the dissimilarity between them. Denoting the matrix determinant by $\det$, the respective sample covariance matrices by $S_{p_i}$ and $S_{q_i}$, and setting $\bar{S}=(S_{p_i}+S_{q_i})/2$, it is defined by 
    \begin{equation} \label{eq:multivariate_DB}
    \textnormal{D}_B \big((p_i),(q_i) \big) = \frac{1}{8}(\overline{p}-\overline{q})^T\bar{S}^{-1}(\overline{p}-\overline{q}) + \frac{1}{2}\ln\left(\frac{\det(\bar{S})}{\sqrt{\det(S_{p_i}) \det(S_{q_i})}}\right).
    \end{equation}
    Note that the first summand (without the constant factor) differs from Hotelling's $T^2$ statistic only by the use of the averaged covariance matrix instead of the pooled one.
    
    A fundamental property of both the Hotelling's $T^2$ statistic and the Bhattacharyya distance is invariance under data translations (e.g., $t^2\big((p_i+v),(q_i+v)\big) = t^2 \big((p_i),(q_i) \big)$ for all $v \in \mathbb{R}^d$). It is thus highly desirable that generalizations of these quantities also exhibit this invariance property.
    We achieve this for Lie groups by considering them as affine manifolds. 
    
    \subsection{Affine Manifolds and Lie Groups}
    
    In the following, we summarize the relevant facts of affine manifolds and (finite-dimensional) Lie group theory; we also explain how the Riemannian is integrated in the affine setting. For more information on Lie groups see~\cite{Helgason2001} and \cite{Postnikov2013}; the latter is also a good reference on affine manifolds. Additional information on differential geometry can, for example, be found in~\cite{doCarmo1992}. In the following we use ``smooth'' synonymously with ``infinitely  often differentiable''. 
    
    Let $M$ be a smooth manifold. We denote the tangent space at $p \in M$ by $T_pM$, and the sets of smooth functions and vector fields on $M$ by $\mathscr{C}^{\infty}(M)$ and $\Gamma(TM)$, respectively. If $X \in \Gamma(TM)$, then $X_p$ denotes its value in $T_pM$.
    Given $X,Y \in \Gamma(TM)$, an affine connection $\nabla$ yields a way to differentiate $Y$ along $X$; the result $\nabla_X Y \in \Gamma(TM)$ being again a vector field. 
    If $M$ is \textit{additionally} endowed with a Riemannian metric (i.e., a smoothly varying inner product on the tangent spaces), then (for every such metric) there is a unique affine connection called Levi-Civita connection. 
    Whenever we speak of a \textit{Riemannian} manifold in the following, we mean a smooth manifold that is endowed with the Levi-Civita connection of some arbitrary Riemannian metric.
    
    With $\gamma' = \frac{\ddd \gamma}{\ddd t}$, we can define a geodesic $\gamma: [0,1] \to M$ by $\nabla_{\gamma'} \gamma' = 0$ as a curve without acceleration.
    An important fact is that every point $p \in M$ has a so-called normal convex neighborhood $U$. Each pair $q, r \in U$ can be joined by a unique geodesic $[0,1] \ni t \mapsto \gamma(t;q,r)$ that lies completely in $U$. 
    Furthermore, with $\gamma'(0;p,q) = v$, this defines the \textit{exponential} $\Exp_p: T_pM \to M$ \textit{of the connection} $\nabla$ at $p$ by 
    $$\Exp_p(v) = \gamma(1;p,q).$$ 
    It is a local diffeomorphism with local inverse 
    $$\Log_p(q) = \gamma'(0;p,q).$$ 
    In the case of Riemannian manifolds, $\Exp$ and $\Log$ are called \textit{Riemannian exponential} and \textit{logarithm}, respectively. 
    \bigskip

    A Lie group $G$ is a smooth manifold that (in addition) has a compatible group structure; that is, there is a smooth (not necessarily commutative) group operation $G \times G \ni (g,h) \mapsto gh \in G$ with corresponding identity element $e \in G$ such that the inversion map $g \mapsto g^{-1}$ is also smooth. 
    Vector spaces (with addition) are instances of Lie groups.
    Another example is the general linear group GL($n$), that is, the set of all bijective linear mappings on an $n$-dimensional vector space $V$, where the group operation is the composition of mappings (i.e., a matrix multiplication), with $e$ being the identity map.
    Whenever we speak of matrix groups in the following, arbitrary closed subgroups of GL($n$) are meant.

    For each $g \in G$ the group operation defines two automorphisms on $G$: the \emph{left} and \emph{right translation} $L_g: h \mapsto gh$ and $R_g: h \mapsto hg$. 
    Their derivatives $d_hL_g$  and $d_hR_g$ at $h \in G$ map tangent vectors $v \in T_hG$ bijectively to the tangent spaces $T_{gh}G$ and $T_{hg}G$, respectively. In particular, it holds that 
    $$T_gG = \{\dd_eL_g(v): v \in T_eG \} = \{ \dd_eR_g(w) : w \in T_eG\}.$$ 
    Thus, each $X_e \in T_eG$ determines a smooth vector field $X \in \Gamma(TG)$ by $X_g = \dd_eL_g(X_e)$ for all $g \in G$. 
    It is called left invariant because $X_{L_g(h)} = \dd_hL_g(X_h)$ for all $h \in G$, that is, the value at a left translated point is the left translated vector. 
    (For matrix groups with identity matrix $I$ we get the simple equation $\dd_IL_A(M) = AM$ for an element $A$ and a matrix $M$ in the tangent space at $I$.) 
    Together with the Lie bracket they form the so-called Lie algebra $\mathfrak{g} \subset \Gamma(TG)$ of $G$.  
    (Remember that the \textit{Lie bracket} of $X, Y \in \Gamma(TG)$ is defined by $[X,Y] = XY - YX \in \Gamma(TG)$\footnote{This is to be understood as a differential operator, i.e., if $\psi \in \mathscr{C}^{\infty}(G)$, then $[X,Y]\psi \in \mathscr{C}^{\infty}(G)$, where $([X,Y]\psi)(g)$ measures the failure to commute when taking directional derivatives of $\psi$ at $g$ in directions $X_g$ and $Y_g$.}.)
    Furthermore, the converse also holds: every left invariant vector field is uniquely determined by its value at the identity. Consequently, $T_eG$ and $\mathfrak{g}$ are isomorphic\footnote{To be exact, $T_eG$ with the bracket $[X_e,Y_e] := [X,Y]_e$ is isomorphic to $\mathfrak{g}$ with its Lie bracket.}.

    Since $d_eL_g$ is bijective for every $g \in G$ (and $L_g$ smooth in $g$), any basis of $T_eG$ determines a unique global frame (a smooth choice of basis for each tangent space) of left invariant vector fields. In particular, for each $g \in G$ any basis of $T_eG$ can be smoothly transported to a basis of $T_gG$ without dependence on the ``path'' (``take the basis from the corresponding global frame''). Because of this, we can view $T_eG$ as the \textit{reference tangent space} of $G$.
    
    Of course, right invariant vector fields are defined analogously and have parallel properties to left invariant fields. 
    
    The integral curve $\alpha_{X}: \mathbb{R} \to G$ through $e$ of an invariant (left or right) vector field $X \in \Gamma(TG)$ determines a unique 1-parameter subgroup of $G$ since $\alpha_X(s + t) = \alpha_X(s) \alpha_X(t)$ for all $s,t \in \mathbb{R}$. 
    The \textit{group exponential} $\exp: T_eG \to G$ is then defined by 
    $$\exp(X_e) = \alpha_X(1).$$ 
    It is also a diffeomorphism in a neighborhood $V$ of $e$ and, hence, we can define the \textit{group logarithm} $\log$ as its inverse there. In the case of matrix groups they coincide with the matrix exponential and logarithm. Important for us will be the inverse consistency of the group logarithm, that is, for all $g \in V$
    \begin{equation} \label{eq:inverse_consistency}
        \log(g^{-1}) = -\log(g).
    \end{equation}
    
    Central to our constructions will be the \textit{canonical Cartan-Shouten (CCS) connection} $\nabla$ of $G$. On $\mathfrak{g}$, it is defined by
    $$\nabla_XY = \frac{1}{2} [X,Y], \quad X,Y \in \mathfrak{g};$$
    see~\cite[Ch.\ 6]{Postnikov2013}. We can extend it to general vector fields, since $\mathfrak{g}$ provides global frames.
    If we endow $G$ with its CCS connection, then geodesics and left (or right) translated 1-parameter subgroups coincide; that is, for any geodesic $\gamma$ in $G$ with $\gamma(0) = g$ both $t \mapsto g^{-1}\gamma(t)$ and $t \mapsto \gamma(t)g^{-1}$ are 1-parameter subgroups of $G$. 
    Thus, on a Lie group $G$ endowed with the CCS connection, the local exponential and logarithm of the connection are given by
    \begin{align} 
        \Exp_g(v) &= g \exp \left(\dd_g L_{g^{-1}}(v) \right) = \exp \left(\dd_g R_{g^{-1}}(v) \right) g, \quad v \in T_gG, \nonumber \\ 
        \Log_g(h) &= \dd_eL_g \log(g^{-1}h) = \dd_eR_g \log(hg^{-1}), \quad h \in U; \label{eq:CCS_logarithm}
    \end{align}
    see~\cite[Cor.\ 5.1]{pennec2019riemannian}. 
    On the other hand, the CCS connection is the Levi-Civita connection of an invariant Riemannian metric only if the latter is bi-invariant, that is, invariant under left \textit{and} right translations~\cite{GhanamHindelehThompson2007}.\footnote{In the published version, ``invariant'' is missing. There are examples of so-called \textit{non-perfect}~\cite[Thm.\ 2.2]{diatta2024dual} Lie groups for which there is a Riemannian metric that is neither left nor right invariant but whose Levi Civita connection is the CCS connection~\cite[Rem.\ 2.1]{diatta2024dual}.} This is the only case in which $\Exp_e \equiv \exp$ and $\Log_e \equiv \log$ on a Lie group with invariant Riemannian structure.
    Unfortunately, there are many Lie groups that do not possess a bi-invariant metric. A well-known example is the group of rigid body transformations SE($3$) (or, more generally, the special Euclidean group SE($n$) for $n>1$). Indeed, any Lie group that is not a direct product of compact and commutative groups does not have have a bi-invariant metric~\cite{MiolanePennec2015}, including $\textnormal{GL}(n)$, $\Glp[n]$, the special linear group $\textnormal{SL}(n)$ (i.e., all linear, orientation and volume preserving transformations of $\mathbb{R}^n$), and the Heisenberg group~\cite{PennecLorenzi2020} (again, all for $n>1$).
    
    Another fundamental automorphism of $G$ is the conjugation $C_g: h \mapsto ghg^{-1}$. Important to us will be its differential at the identity, which we call \emph{group adjoint} and denote by $\Ad(g)$. It acts bijectively on vectors $v \in T_eG$ by 
    $$\Ad(g)v = \dd_{g^{-1}}L_g (\dd_{e}R_{g^{-1}}(v)) = \dd_gR_{g^{-1}} (\dd_{e}L_{g}(v)).$$
    For matrix groups this reduces to $\Ad(A)(M) = AMA^{-1}$ for elements $A$ and matrices $M$ in the tangent space at the identity. 
    The group adjoint yields the following crucial relation~\cite[Thm.\ 6]{PennecArsigny2013}: For $f,g \in G$ such that $\log(fg^{-1})$ exists, it links logarithms of left and right translated points according to
    \begin{equation}\label{eq:left_right}
        \log(gf^{-1}) = \Ad(f) \log(f^{-1}g).
    \end{equation}

\subsection{The Group Mean} 

    The fact that bi-invariant metrics generally do not exist on Lie groups makes the construction of (statistical) notions that are invariant/equivariant under translations from both left and right very difficult in the Riemannian setting. 
    On the other hand, Lie groups that are endowed with their CCS connections possess the natural prerequisites.
    Following Pennec and Arsigny, we define the mean in a Lie group as an exponential barycenter~\cite[Def.\ 3]{PennecArsigny2013}.
    \begin{definition}[Group mean]
        Let $G$ be a Lie group endowed with its CCS connection and let $g_1,\dots,g_m \in V \subseteq G$, $V$ being simply-connected, such that $\log(g^{-1}g_i)$ exists for all $g \in V$ and $i=1,\dots,m$. Then, we call $\overline{g} \in G$ \textnormal{group mean of $g_1,\dots,g_m$} if
        \begin{equation} \label{eq:baryceneter}
        \sum_{i=1}^m \Log_{\overline{g}}(g_i) = 0.
    \end{equation}    
    \end{definition}
    Note that, because of \cref{eq:CCS_logarithm}, equation \cref{eq:baryceneter} is equivalent to both
    \begin{equation*}
        \sum_{i=1}^m \log(\overline{g}^{-1}g_i) = 0 \quad \text{and} \quad \sum_{i=1}^m \log(g_i\overline{g}^{-1}) = 0.
    \end{equation*}
    From this we can see that the elements $g_1,\dots,g_m$ must be ``sufficiently localized'' to obtain a mean value; because if they are too far apart, the logarithm may not be defined.
    The most general result known so far on existence and uniqueness is the following:
    If $U \subseteq G$ is a CSLCG (convex with semilocal convex geometry) neighborhood and $g_1,\dots,g_m \in U$, then their group mean exists and is unique. The details can be found in \cref{app:existence_uniqueness}.
    
    Central properties of the group mean are summarized in the following theorem \cite[Thm.\ 5.13]{PennecLorenzi2020}.
    \begin{theorem}[Equivariance of the group mean] \label{thm:mean_equivariance}
        Let $G$ be a Lie group endowed with its CCS connection and $\overline{g}$ be a group mean of $g_1,\dots,g_m \in V \subseteq G$. Then, for any $f \in G$, the group means of the left translated data $(fg_1,\dots,fg_m)$, right translated data $(g_1f,\dots,g_mf)$ and inverted data $(g_1^{-1},\dots,g_m^{-1})$ are $f\overline{g}$, $\overline{g}f$ and $\overline{g}^{-1}$, respectively.
    \end{theorem}
    According to the theorem, group means are equivariant under the group operations. These are very favorable properties that make (statistical) data analysis more robust.
    
    Finally, group means can be computed efficiently with a fixed point iteration~\cite[Thm.\ 5.14]{PennecLorenzi2020}.
    
\subsection{Sample Covariance and Bi-invariant Mahalanobis Distance}

    In \cite{PennecArsigny2013}, Pennec and Arsigny also define the 
    sample covariance of Lie-group-valued data $(g_1,\dots,g_m)$ with group mean $\overline{g}$. Denoting the vector of coordinates of a tangent vector as well as the representing matrix of a (bi)linear map in any given basis by $[\cdot]$, it is the (2,0)-tensor
    $$\Sigma^*_{g_i} = \frac{1}{m} \sum_{l=1}^m \Log_{\overline{g}}(g_l) \otimes \Log_{\overline{g}}(g_l) \in T_{\overline{g}}G \otimes T_{\overline{g}}G,$$
    where the tensor product $\otimes$ means that in any basis of $T_{\overline{g}}G$, the representing matrix is given by $[\Sigma^*_{g_i}] = 1/m \sum_l [\Log_{\overline{g}}(g_l)] [\Log_{\overline{g}}(g_l)]^T$.
    From this, they define the \textit{bi-invariant Mahalanobis distance} of $f \in G$ to the distribution of the $g_i$ by
    \begin{equation} \label{eq:mahalanobis_distance}
        \mu^2_{(\overline{g}, \Sigma^*_{g_i})}(f) =  \Big[\Log_{\overline{g}}(f) \Big]^T \Big[\Sigma_{g_i}^{*-1} \Big] \Big[\Log_{\overline{g}}(f) \Big],
    \end{equation}
    if $\Log_{\overline{g}}(f) = \log(\overline{g}^{-1}f)$ exists.
    (Here, $\Sigma_{g_i}^{*-1}$ denotes the inverse of $\Sigma^*_{g_i}$.)
    The distance~\cref{eq:mahalanobis_distance} is left and right invariant, that is, 
    $$\mu^2_{(\overline{g}, \Sigma^*_{g_i})}(f) = \mu^2_{(h\overline{g}, \Sigma^*_{hg_i})}(hf) = \mu^2_{(\overline{g}h, \Sigma^*_{g_ih})}(fh)$$
    for all $h \in G$ and all $f$ such that $\Log_{\overline{g}}(f)$ exists; see~\cite[p.\ 181]{PennecLorenzi2020}.

\section{Bi-invariant Dissimilarity Measures for Sample Distributions on Lie Groups}
    In this section we derive bi-invariant extensions of the Hotelling $T^2$ statistic and Bhattacharyya distance to a Lie group $G$ that is endowed with its CCS connection.
    To this end, we always consider a CSLCG neighborhood $U$ of a Lie group $G$.

\subsection{Extending the Hotelling \texorpdfstring{{\boldmath$T^2$}}{T Squared} Statistic}
    To generalize Hotelling's $T^2$ statistic, we define a slightly modified sample covariance.
    This is motivated by the fact that $T_eG$ is the reference tangent space for the whole group.
    \begin{definition}[Centralized sample covariance] \label{def:centralized_covariance}
    Given data $(g_1,\dots,g_m)$ in $U$ with group mean $\overline{g}$, its \textnormal{left-centralized sample covariance} is defined as the (2,0)-tensor 
    $$\Sigma_{g_i} = \frac{1}{m} \sum_{l=1}^m \log(\overline{g}^{-1}g_l) \otimes \log(\overline{g}^{-1}g_l) \in T_eG \otimes T_eG,$$
    which in local coordinates is 
        $$ \Big[\Sigma_{g_i} \Big] = \frac{1}{m} \sum^m_{l=1} \Big[\log(\overline{g}^{-1}g_l) \Big] \Big[\log(\overline{g}^{-1}g_l) \Big]^T.$$
    \end{definition}
    Since we only use the centralized covariance that is defined under left translations of the data, we will drop the word ``left'' in the following.
    The centralized covariance will be fundamental for our constructions since it allows to compare the variability of data sets with different means in a single tangent space; this, in turn, makes it possible to transfer notions from multivariate statistics.
    Taking advantage of the bi-invariance property we also have
    \begin{equation*}
        \mu^2_{(\overline{g}, \Sigma^*_{g_i})}(f) = \mu^2_{(e, \Sigma_{g_i})}(\overline{g}^{-1}f).
    \end{equation*}
    Moreover, we can now define the pooled covariance at the identity: 
    \begin{definition}[Pooled sample covariance]
        Given data sets $(g_1,\dots,g_m)$ and $(h_1,\dots,h_n)$ in $U$ with group means $\overline{g}$ and $\overline{h}$, their \textnormal{left-pooled sample covariance} is defined by
        $$\widehat{\Sigma}_{g_i,h_i} = \frac{1}{m + n - 2} \left( m\Sigma_{g_i} + n\Sigma_{h_i} \right).$$
    \end{definition}
    As for the centralized covariance, we drop the word ``left'' for the pooled covariance. Furthermore, because there is no danger of confusion, we only write $\widehat{\Sigma}$ from now on. 
    
    In order to obtain bi-invariant indices from the newly-defined notions, we need to understand how they transform under joint translations of the data. For this, we denote the centralized covariance of data that was jointly left or right translated with $f \in G$ by $f \bullet \Sigma_{g_i}$ and $\Sigma_{g_i} \bullet f$, respectively, (i.e., $f \bullet \Sigma_{g_i} = \Sigma_{fg_i}$ and $\Sigma_{g_i} \bullet f = \Sigma_{g_if}$); furthermore, $\textnormal{inv} \bullet \Sigma_{g_i} = \Sigma_{g_i^{-1}}$ denotes the centralized covariance of the inverted data points. We also extend these definitions linearly to weighted sums of covariances. 
    Then, we have the following lemma.
    \begin{lemma} \label{lem:centralized_covariance}
         Given data $(g_1,\dots,g_m)$ in $U$, its centralized covariance is invariant under left translations, while under right translations with $f \in G$ it transforms according to
        $$\Big[\Sigma_{g_i} \bullet f \Big] = \Big[\Ad(f^{-1}) \Big] \Big[\Sigma_{g_i} \Big] \Big[\Ad(f^{-1}) \Big]^T.$$
        Furthermore, if the data points are inverted, then
        $$\Big[\inv \bullet \Sigma_{g_i} \Big] = \Big[\Ad(\overline{g}) \Big] \Big[\Sigma_{g_i} \Big] \Big[\Ad(\overline{g}) \Big]^T.$$
    \end{lemma}
    \begin{proof}
        Using the equivariance of the mean gives
        \begin{equation*}
            \Big[f \bullet \Sigma_{g_i} \Big] = \frac{1}{m} \sum^m_{l=1} \Big[\log \left(\overline{g}^{-1}f^{-1}fg_l \right) \Big] \Big[\log \left(\overline{g}^{-1}f^{-1}fg_l \right)]^T = \Big[\Sigma_{g_i} \Big],
        \end{equation*}
        
        and applying \cref{eq:left_right} yields 
        \begin{align*}
            \Big[\Sigma_{g_i} \bullet f \Big] &= \frac{1}{m} \sum^m_{l=1} \Big[\log \left(f^{-1}\overline{g}^{-1}g_lf \right) \Big] \Big[\log \left(f^{-1}\overline{g}^{-1}g_lf \right) \Big]^T \\
            &= \frac{1}{m}  \sum^m_{l=1} \Big[\Ad \left(f^{-1} \right) \Big] \Big[\log \left(\overline{g}^{-1}g_l \right) \Big] \Big[\log \left(\overline{g}^{-1}g_l \right)]^T \Big[\Ad \left(f^{-1} \right) \Big]^T \\
            &= \Big[\Ad \left(f^{-1} \right) \Big] \Big[\Sigma_{g_i} \Big] \Big[\Ad \left(f^{-1} \right) \Big]^T.
        \end{align*}
        
        After inverting the data points we find, using \cref{thm:mean_equivariance}, \cref{eq:inverse_consistency}, and \cref{eq:left_right}, that
        \begin{align*}
            \Big[\inv \bullet \Sigma_{g_i} \Big] &= \frac{1}{m} \sum^m_{l=1} \Big[\log \left(\overline{g} g_l^{-1} \right) \Big] \Big[\log \left(\overline{g} g_l^{-1} \right) \Big]^T \\
            &= \frac{1}{m} \sum^m_{l=1} (-1)^2 \Big[\log \left(g_l\overline{g}^{-1} \right) \Big] \Big[\log \left(g_l\overline{g}^{-1} \right) \Big]^T \\
            &= \Big[\Ad(\overline{g}) \Big] \Big[\Sigma_{g_i} \Big] \Big[\Ad(\overline{g}
            ) \Big]^T.
        \end{align*}
    \end{proof}
    We can translate this directly to the pooled covariance:
    \begin{corollary} \label{cor:pooled_covariance}
        The pooled covariance is invariant under left translations of the data and transforms under right translations with $f \in G$ according to
        $$\Big[\widehat{\Sigma} \bullet f \Big] = \Big[\Ad \left(f^{-1} \right) \Big] \Big[\widehat{\Sigma} \Big] \Big[\Ad \left(f^{-1} \right) \Big]^T.$$
        If the data is inverted, then we have
        $$\Big[ \inv \bullet \widehat{\Sigma} \Big] = \frac{1}{m + n - 2} \left( m \Big[\Ad(\overline{g}) \Big] \Big[\Sigma_{g_i} \Big] \Big[\Ad(\overline{g}) \Big]^T + n \Big[\Ad \left(\overline{h} \right) \Big] \Big[\Sigma_{h_i} \Big] \Big[\Ad \left(\overline{h} \right) \Big]^T \right).$$
        In particular, if $\overline{g} = \overline{h}$, then
        $$\Big[ \inv \bullet \widehat{\Sigma} \Big] = \Big[\Ad(\overline{g}) \Big] \Big[ \widehat{\Sigma} \Big] \Big[\Ad(\overline{g}) \Big]^T.$$
    \end{corollary}
    With this, we propose the following generalization of the Hotelling $T^2$ statistic for Lie groups, using the abbreviated notation $\big((g_i),(h_i) \big)$ for $\big((g_1,\dots,g_m),(h_1,\dots,h_n) \big)$:
    \begin{definition}[Bi-invariant Hotelling $T^2$ statistic] \label{def:t2_statistic}
        Let $(g_1,\dots,g_m)$ and $(h_1,\dots,h_n)$ be data sets in $U$ with group means $\overline{g}$ and $\overline{h}$ such that $\log(\overline{g}^{-1}\overline{h})$ exists, and let $\Sigma$ be their pooled sample covariance. Then, the \textnormal{bi-invariant Hotelling $T^2$ statistic} is defined by 
        \begin{align*}
            t^{2} \big((g_i),(h_i) \big) &= \frac{mn}{m + n} \mu^2_{(e, \widehat{\Sigma})} \left( \overline{g}^{-1}\overline{h} \right).
        \end{align*}
    \end{definition}
    
    This definition is independent of the chosen basis of $T_eG$ \cite[p.\ 39]{PennecArsigny2013}. Furthermore, it reduces to \cref{eq:Euclidean_test} when $G = \mathbb{R}^n$. This is a consequence of the existence of a common reference tangent space; as the latter is missing in Riemannian manifolds, reduction to the multivariate case is problematic using the Riemannian approach. Indeed, the only generalization of Hotelling's $T^2$ statistic to Riemannian manifolds known to the authors does not have this property; see~\cref{app:riemannian} for more on this. 
    
    The following theorem shows that the newly defined notion has the invariance properties we aimed for and does not depend on the order of the data sets.
    \begin{theorem}[Properties of the bi-invariant Hotelling $T^2$ statistic] \label{thm:hotelling}
        Let $(g_1,\dots,g_m)$ and $(h_1,\dots,h_n)$ be data sets in $U$ with group means $\overline{g}$ and $\overline{h}$, respectively.
        The bi-invariant Hotelling $T^2$ statistic has the following properties: \\
        (i)\ \ \ it is symmetric, \\
        (ii)\ \  it is invariant under left and right translations, \\
        (iii) if $\overline{g} = \overline{h}$, then it is invariant under inversion.
    \end{theorem}
    \begin{proof}
        Note first that the pooled covariance is independent of the order of the data groups. Thus, using \cref{eq:inverse_consistency}, property (i) follows from
        \begin{align*}
        t^{2} \big((g_i),(h_i) \big) = \frac{mn}{m + n} (-1)^{2} \Big[\log \left(\overline{h}^{-1}\overline{g} \right) \Big]^T \Big[\widehat{\Sigma} \Big]^{-1} \Big[\log \left(\overline{h}^{-1}\overline{g} \right) \Big] = t^{2}\big((h_i),(g_i) \big).
        \end{align*}
        
        Next, we show (ii). Invariance under left translations follows directly since $f$ and $f^{-1}$ cancel out as in the proof of \cref{lem:centralized_covariance}. Let $f \in G$. 
        Then, using \cref{thm:mean_equivariance,eq:left_right,cor:pooled_covariance}, we find
        \begin{align*}
            t^{2}\big((g_if),(h_i f) \big) &= \frac{mn}{m + n}   \Big[\log \left(f^{-1}\overline{g}^{-1}\overline{h} f \right) \Big]^T \Big[\widehat{\Sigma} \bullet f \Big]^{-1} \Big[\log \left(f^{-1}\overline{g}^{-1}\overline{h} f \right) \Big]  \\
            &= \frac{mn}{m + n} \Big[\log \left(\overline{g}^{-1}\overline{h} \right) \Big]^T \Big[\Ad(f^{-1}) \Big]^T \Big[\Ad(f^{-1}) \Big]^{-T} \Big[\widehat{\Sigma} \Big]^{-1} \\
            &\ \ \ \ \ \ \ \ \ \ \ \ \ \ \ \ \ \ \ \ \ \ \ \Big[\Ad \left(f^{-1} \right) \Big]^{-1} \Big[\Ad \left(f^{-1} \right) \Big] \Big[\log \left(\overline{g}^{-1}\overline{h} \right) \Big] \\
            &= t^{2} \big((g_i),(h_i) \big).
        \end{align*}

        Finally, (iii) follows immediately from $\log(e) = 0$.
    \end{proof}
    Note that although the statistic is not invariant under inversion in general, the fact that equality of means is invariantly detected makes it also interesting (e.g., to perform hypothesis tests for equality of means) when invariance under inversion is of interest. (Nevertheless, the authors are not aware of any such application.)
    \begin{remark}
        Observe that we could replace left by right-centralized sample covariances in all definitions of this section. The Hotelling $T^2$ statistic will then be different in general, but it will have the same invariance properties as~\cref{def:t2_statistic}. 
     \end{remark}

\subsection{Extending the Bhattacharyya Distance} \label{sec:bhattacharyya}

    Since the first summand of the multivariate Bhattacharyya distance \cref{eq:multivariate_DB} coincides with a Hotelling $T^2$ statistic, we can generalize it using the bi-invariant Hotelling $T^2$ statistic.
    To this end, we begin with the definition of the averaged sample covariance:
    \begin{definition}[Averaged sample covariance] \label{def_averaged_covariance}
        Given data sets $(g_1,\dots,g_m)$ and $(h_1,\dots,h_n)$ in $U$ with group means $\overline{g}$ and $\overline{h}$, their \textnormal{left-averaged sample covariance} is defined by
        $$\overline{\Sigma}_{g_i, h_i} = \frac{1}{2} (\Sigma_{g_i} + \Sigma_{h_i}).$$
    \end{definition}
    Again, in the following we drop the word ``left'' and the subscripts. Since it only differs in the weighting, the averaged covariance has the same properties as the pooled covariance.
    \begin{corollary} \label{cor:averaged_covariance}
        The averaged covariance is invariant under left translations of the data and transforms under right translations with $f \in G$ according to
        $$\Big[\overline{\Sigma}\bullet f \Big] = \Big[\Ad \left(f^{-1} \right) \Big] \Big[\overline{\Sigma} \Big] \Big[\Ad \left(f^{-1} \right) \Big]^T.$$
        If the data is inverted, then we have
        $$\Big[ \inv \bullet \overline{\Sigma} \Big] = \frac{1}{2} \left( \Big[\Ad(\overline{g}) \Big] \Big[\Sigma_{g_i} \Big] \Big[\Ad(\overline{g}) \Big] + \Big[\Ad \left(\overline{h} \right) \Big] \Big[\Sigma_{h_i} \Big] \Big[\Ad \left(\overline{h} \right) \Big] \right).$$
        In particular, if $\overline{g} = \overline{h}$, then
        $$\Big[ \inv \bullet \overline{\Sigma} \Big] = \Big[\Ad(\overline{g}) \Big] \Big[ \overline{\Sigma} \Big] \Big[\Ad(\overline{g}) \Big]^T.$$
    \end{corollary}
    With this, we can generalize the Bhattacharyya distance to Lie groups:
    \begin{definition}[Bi-invariant Bhattacharyya distance] \label{def:Bhat_distance}
        Let $(g_1,\dots,g_m)$ and $(h_1,\dots,h_n)$ be data sets in $U$ with group means $\overline{g}$ and $\overline{h}$. Then, the \textnormal{bi-invariant Bhattacharyya distance} is defined by 
        \begin{equation*}
            \textnormal{D}_B \big((g_i), (h_i) \big) = \frac{1}{8} \mu^2_{(e, \overline{\Sigma})} \left( \overline{g}^{-1}\overline{h} \right) + \frac{1}{2}\ln\left(\frac{\det \left(\Big[\overline{\Sigma} \Big] \right)}{\sqrt{\det \left( \Big[\Sigma_{g_i} \Big] \right) \det \left( \Big[\Sigma_{h_i} \Big] \right)}} \right).
        \end{equation*}
    \end{definition}
    The following theorem shows that the attribute ``bi-invariant''  is indeed justified.
    \begin{theorem}[Properties of the bi-invariant Bhattacharyya distance] \label{thm:properties_Bhatt}
        Let $(g_1,\dots,g_m)$ and $(h_1,\dots,h_n)$ data sets in $U$ with group means $\overline{g}, \overline{h}$, respectively.   
        The bi-invariant Bhattacharyya distance has the following properties: \\
        (i)\ \ \ \ it is symmetric, \\
        (ii) \ \ it does not depend on the chosen basis, \\
        (iii)\ \ it is invariant under left and right translations, \\
        (iv) \ \hspace{.5mm}if $\overline{g} = \overline{h}$, then it is invariant under inversion.
    \end{theorem}
    \begin{proof}
        Note that we can focus on the second summand, since for the first all results can be shown as in the proof of \cref{thm:hotelling} (replacing the pooled by the averaged covariance).
        Thus, (i) follows immediately, since $\overline{\Sigma}$ is invariant under an exchange of the data sets.
        
        Next we show (iii). Let $f \in G$ and $\rho_f = \det(\Ad(f^{-1}))$; note that $\rho_f \ne 0$. Invariance of the second summand under left translations follows from \cref{lem:centralized_covariance,cor:averaged_covariance}. They further imply that
        \begin{align*}
            \frac{\det \left( \Big[ \overline{\Sigma} \bullet f \Big] \right)}{\sqrt{\det \left( \Big[\Sigma_{g_i} \bullet f \Big] \right) \det \left( \Big[\Sigma_{h_i} \bullet f \Big] \right)}} &= \frac{\rho_f^2 \det \left( \Big[\overline{\Sigma}\Big] \right)}{\rho_f^2 \sqrt{\det \left( \Big[\Sigma_{g_i} \Big] \right) \det \left( \Big[\Sigma_{h_i} \Big] \right)}},
        \end{align*}
        which yields invariance under right translations.
        The proofs of (ii) and (iv) work analogously (using \cref{cor:averaged_covariance} for (iv) and replacing group adjoints with basis transformations for (ii)). 
    \end{proof}
    
    Note that in \cref{def:Bhat_distance} we take determinants of matrix representations of (2,0)-tensors instead of (0,2)-tensors like in the Euclidean case (for background information see, e.g.,~\cite[Ch.\ 4]{doCarmo1992}). This does not make a difference since lowering indices with any (auxiliary) metric does not have an influence on the \textit{ratio} of the determinants in \cref{def:Bhat_distance}; the additional terms (i.e., the matrix representation of the metric) cancel out analogously to those introduced by translations (i.e., matrix representations of group adjoints) in the proofs of \cref{thm:hotelling,thm:properties_Bhatt}.
    
    \begin{remark}
        A notion from multivariate statistics that is intimately related to the Bhattacharyya distance is the \textit{Hellinger distance}
        $H$~\cite{Hellinger1909} (see also \cite[p.\ 51]{Pardo2005}):
        \begin{equation} \label{eq:hellinger_distance}
            H \big((p_i),(q_i) \big) = \sqrt{1 - 1 /\exp \big(\textnormal{D}_B \big((p_i), (q_i) \big) \big)}
        \end{equation}
        for data sets $(p_1,\dots,p_m)$ and $(q_1,\dots,q_n)$ in $\mathbb{R}^d$. It has found various applications: for example, in decision trees~\cite{Cieslak_ea2011_HellingerDD}, internet telephony~\cite{Sengar_ea2008}, and data visualization~\cite{Rao1995}.
        
        By using the bi-invariant Bhattacharyya distance in~\cref{eq:hellinger_distance}, we directly obtain a \textit{bi-invariant Hellinger distance} on Lie Groups with the same invariance properties.
    \end{remark}

\subsection{Connection of the Bi-invariant Bhattacharyya Distance to Densities}
    In the multivariate setting, the Bhattacharyya distance has a more general integral definition for distributions with probability density function (pdf). We show in this section that the bi-invariant Bhattacharyya distance is compatible with this view. Like in Euclidean space, we obtain it for Gaussian-like distributions from the generalized integral definition. On the other hand, in contrast to our definition, the integral version has the drawback that it is not bi-invariant in general.

    Let $P$ and $Q$ be two random variables on $\mathbb{R}^d$ with pdfs $p$ and $q$, respectively.
    The Bhattacharyya distance then takes the form
    $$\textnormal{D}_B(p,q) = -\ln \left( \int_{\mathbb{R}^d} \sqrt{pq}\, \dd x \right),$$
    which reduces to \cref{eq:multivariate_DB} for two normal distributions. In this section we investigate classes of Lie groups and distributions for which the above (or more precisely an analog) integral reduces to the bi-invariant Bhattacharyya distance.
    
    On a Lie group $G$ of dimension $d$, integrals of functions are defined via differential $d$-forms\footnote{Remember that~\cite[Sec.\ 1.3]{Helgason2001} on a smooth manifold $M$, a differential $d$-form $\omega$ smoothly assigns to each $p \in M$ an alternating $d$-linear map $\omega_p: (T_pM)^d \to \mathbb{R}$. If $\psi \in \mathscr{C}^{\infty}(M)$, then $\psi\, \omega$ is also a $d$-form on $M$.}; see, for example,~\cite[Ch.\ 26]{Postnikov2013}. For our purpose, it suffices to recount how one integrates over the domain of a single chart. (In general, integrals are defined via so-called partitions of unity.)
    
    Let $\Phi: V \to U$ be a diffeomorphism between $V \subseteq \mathbb{R}^d$ and a neighborhood $U \subseteq G$ (i.e., an inverse chart) and $\omega$ a $d$-form on $U$. Then, there is a unique $\varphi \in \mathscr{C}^{\infty}(\Phi(V))$ such that the pullback\footnote{Let $F:M \to N$ be a diffeomorphism of two smooth $d$-dimensional manifolds $M, N$ and $\omega$ a $d$-form on $N$. The pullback $F^*(\omega)$ of $\omega$ along $F$ is a $d$-form on $M$ defined by $F^*(\omega)(X_1,\dots,X_d) = \omega(\dd F(X_1),\dots,\ \dd F(X_d))$ for $X_1,\dots,X_d \in \Gamma(TM)$. Note that $F^*(\psi\, \omega) = (\psi \circ F)\, F^*(\omega)$ for all $\psi \in \mathscr{C}^{\infty}(N)$.} of $\omega$ along the diffeomorphism $\Phi$ is given by $\Phi^*(\omega) = \varphi\, \dd x$ ($\dd x$ being the ordinary volume form on $\mathbb{R}^d$). The integral of $\omega$ over $U$ is then defined by
    $$ \int_U \omega = \int_V \Phi^*(\omega) = \int_V \varphi\, \dd x$$
    if the right-hand side exists.
    
    In the affine setting there is no canonical volume form on $G$. Instead, we must choose a reference measure, that is, a nonzero left invariant $d$-form $\dd g$ on $G$ (a \textit{left Haar measure}) consistent with the orientation of $G$ (which we choose such that $\exp$ is orientation preserving). Luckily, a different choice of reference measure only leads to multiplying each integral with the same constant, which is why the actual choice of $\dd g$ depends (up to considerations of machine precision when working on a computer) only on one's taste. Thus, from now on we assume that $\dd g$ is arbitrary but fixed. Then, the integral of a function $\psi \in \mathscr{C}^{\infty}(U)$ is defined by 
    \begin{equation} \label{eq:integral_Lie_grou}
        \int_U \psi\, \dd g.
    \end{equation}
    
    This, in turn, allows to define pdfs on $G$ as (not necessarily smooth) functions $\mu \ge 0$ on $G$ such that $\int_G \mu\, \dd g = 1$ (see, e.g., \cite{Chirikjian2012}).
    Consequently, we get the integral version of the Bhattacharyya distance on Lie groups:
    \begin{equation} \label{def:integral_Bhattacharyya}
        \textnormal{D}_B^{\textnormal{int}}(\mu,\nu) = - \ln \left( \int_G \sqrt{\mu \nu}\, \dd g \right)
    \end{equation}
    for pdfs $\mu, \nu$ on $G$. 
    
    Note that, in general, \cref{eq:integral_Lie_grou} is left invariant (i.e., $\int_G \psi \circ L_h\, \dd g = \int_G \psi\, \dd g$ for each $h \in G$ and $\psi \in \mathscr{C}^{\infty}(G)$ whose integral exists) but not right invariant.\footnote{We have a similar problem when we start with a right-invariant volume form.} 
    Thus, in contrast to~\cref{def:Bhat_distance}, \cref{def:integral_Bhattacharyya} cannot generally be bi-invariant. For this to hold, we need a \textit{unimodular} Lie group, that is, a Lie group with bi-invariant volume form. (A well-known example of a non-unimodular Lie group is the group of invertible affine transformations of $\mathbb{R}^n$.) Therefore, we will only be interested in the unimodular case. Since then distributions can be ``moved around via translations'' without changing integrals, we can consider w.l.o.g.\ those centered at $e$. 
    
    In the following, let $E: \mathbb{R}^d \to T_eG$ be any orientation-preserving vector space isomorphism (which we obtain from choosing a basis of $T_eG$).
    Further, let $\mathfrak{e} \subseteq T_eG$ be a neighborhood of $0 \in T_eG$ such that \textit{first} $\exp|_{\mathfrak{e}}$ is a diffeomorphism onto the neighborhood $W=\exp(\mathfrak{e})$ of $e$, and \textit{second} there is an orientation-preserving diffeomorphism $\Phi : T_eG \to \mathfrak{e}$ between the whole of $T_eG$ and $\mathfrak{e}$. \cref{lem:cslcg_diff_Rd} shows that we can always choose $\mathfrak{e} := \log(W)$ where $W$ is any CSLCG neighborhood of $e$. 
    Now, we can define $\widetilde{\exp} = \exp \circ \Phi \circ E$ and $\widetilde{\log} = E^{-1} \circ \Phi^{-1} \circ \log$. Note that the latter is a chart of $W$ with image $\mathbb{R}^d$. 
    In order to identify densities on $G$ that connect both versions of the Bhattacharyya distance we first need the lemma below.
    \begin{lemma} \label{lem:push_forward_distribution}
        Let $G$ be a $d$-dimensional Lie group and $g \in G$. Further, let $p,q$ be pdfs on $\mathbb{R}^d$ and $\widetilde{\exp}^*(\dd g\big|_W) = \phi\, \dd x$ with $0 < \phi \in \mathscr{C}^{\infty}(\mathbb{R}^d)$. 
        Define 
        \begin{equation} \label{eq:mu}
        \mu(g) = \begin{cases} (p/\phi \circ \widetilde{\log}) (g), &  g \in W, \\ 0, & g \in G \setminus W, \end{cases}
        \end{equation}
        and
        \begin{equation} \label{eq:nu}
        \nu(g) = \begin{cases} (q/\phi \circ \widetilde{\log}) (g), &  g \in W, \\ 0, & g \in G \setminus W. \end{cases}
        \end{equation}
        Then, $\mu$ and $\nu$ are pdfs on $G$ with
        $$ \displaystyle \int_G \sqrt{\mu \nu}\, \dd g = \int_{\mathbb{R}^d} \sqrt{pq}\, \dd x.$$
    \end{lemma}
    \begin{proof}
        Clearly, $\mu, \nu \ge 0$. Furthermore, since $\mu$ and $\nu$ are only supported in $W$ and $\widetilde{\log} \circ \widetilde{\exp} = \text{Id}_{\mathbb{R}^d}$, we find $\int_G \mu\, \dd g = \int_G \nu\, \dd g = 1$ and 
        \begin{align*}
            \int_G \sqrt{\mu \nu}\, \dd g &= \int_W \sqrt{\mu \nu}\, \dd g = \int_{\mathbb{R}^d} \widetilde{\exp}^*(\sqrt{\mu \nu}\, \dd g) =  \int_{\mathbb{R}^d} (\sqrt{\mu \nu} \circ \widetilde{\exp}) \phi\, \dd x =
            \int_{\mathbb{R}^d} \sqrt{p q}\, \dd x.
        \end{align*}
    \end{proof}
    We can write the function $\phi$ also more explicitly. Denoting $\dd y_i = \widetilde{\log}^*(\dd x_i)$, $i=1,\dots,d$, there is $0 < \psi \in \mathscr{C}^{\infty}(G)$ such that $\dd g \big|_W = \psi\, \dd y_1 \wedge \dots \wedge \dd y_d$\footnote{Here, $\wedge$ denotes the ``wedge product'' of differential forms. Any form of maximal degree $d$ can be written as the wedge product of $d$ coordinate 1-forms multiplied by a smooth function; see, e.g.,~\cite[Sec.\ 1.3]{Helgason2001}.}. Hence, the transformation rule for forms of maximal degree (i.e., the dimension of the underlying manifold) implies
    \begin{equation} \label{eq:pullback_Haar_measure}
        \widetilde{\exp}^*(\dd g \big|_W) = \det(\dd\, \widetilde{\exp}) (\psi \circ \widetilde{\exp})\, \dd x_1 \wedge \dots \wedge \dd x_d
    \end{equation}
    and, thus, $\phi = \det(\dd\, \widetilde{\exp}) (\psi \circ \widetilde{\exp})$.
    
    In the following, $\mathcal{N}(x,[\Sigma])$ denotes the multivariate normal distribution on $\mathbb{R}^d$ with mean $x$ and covariance matrix $[\Sigma]$. The following proposition establishes a connection between the bi-invariant Bhattacharyya distance and the integral version.
    \begin{proposition}
        Let $G$ be a $d$-dimensional unimodular Lie group with CCS connection, and let $p$ be the pdf of $\mathcal{N}(0,[\Sigma_{g_i}])$ and $q$ the pdf of $\mathcal{N}([\log(\overline{g}^{-1}\overline{h})],[\Sigma_{h_i}])$ on $\mathbb{R}^d$.
        From $p$ and $q$, define pdfs $\mu$ and $\nu$ on $G$ by \cref{eq:mu,eq:nu}. Further, let $(g_1,\dots,g_m)$ and $(h_1,\dots,h_n)$ be data sets in a CSLCG neighborhood $U \subseteq \exp(\mathfrak{e}) \subseteq G$ with group means $\overline{g}$ and $\overline{h}$, respectively.
        Then, 
        $$\textnormal{D}_B^{\textnormal{int}}(\mu,\nu) = \textnormal{D}_B \big((g_i), (h_i) \big).$$
    \end{proposition}
    \begin{proof}
    The assertion follows from~\cref{lem:push_forward_distribution} and the fact that the Bhattacharyya distance between two multivariate normal distributions is of the form~\cref{eq:multivariate_DB}.
    \end{proof}
    According to the proposition, ``scaled push-forwards'' of the multivariate normal distribution provide a connection between both versions of the Bhattacharyya distance. Densities connected to the these were investigated for SE($n$) in~\cite{ChevallierGuigui2020}---also with an emphasis on bi-invariance. Note that the requirement $U \subseteq \exp(\mathfrak{e})$ is not needed in the proof but is introduced to ensure that the data samples have non-vanishing probabilities under $p$ and $q$. If one wants to choose $p$ and $q$, for example, to build a model, setting $\mathfrak{e} = \log(U)$ seems sensible because thereby the locality condition that we impose on the data for $\textnormal{D}_B$ is met. On the other hand, larger domains could be used if ``outliers'' should be possible.
    An interesting special case are Lie groups whose exponentials are \emph{global} diffeomorphisms. Examples are simply connected, nilpotent Lie groups (like the Heisenberg group, which allows for a unique mean in the entire group~\cite[pp.\ 198--199]{pennec2019riemannian}); they are unimodular and the pullback of the Lebesgue measure under the logarithm is a bi-invariant Haar measure~\cite{Christ1992}. Thus, in a simply-connected, nilpotent Lie group, $\dd g$ can be chosen such that $\psi \equiv 1$ in \cref{eq:pullback_Haar_measure}.
    
\section{Experiments}
    In this section, we use the proposed indices for hypothesis tests in two application fields. Although we have chosen two specific contexts, the testing procedures described below can be used in many other application fields. 
    
    \subsection{Configurations of the Knee Joint under Osteoarthritis} \label{sec:knee_experiment}
        In the first part, we investigate configurations of the human knee joint under osteoarthritis (OA). The relative position of the femur with respect to the tibia is described by a rigid-body transformation. Using both the Hotelling $T^2$ statistic and Bhattacharyya distance in a two-sample test, we infer a significant difference between knee configurations of people with severe OA and healthy controls; we thus find the well-known joint space narrowing that is indicative for OA.
    
    \subsubsection{Data Description}
        The first data set, which was also used in~\cite{vonTycowicz_ea2018}, is derived from the Osteoarthritis
        Initiative\footnote{\url{https://nda.nih.gov/oai/}} (OAI). The OAI is a longitudinal study of knee osteoarthritis that provides (among others) publicly accessible clinical evaluation data and radiological images from 4,796 men and women of age 45--79. By means of a special support plate it was made sure that the patients' legs were constrained to the same posture. From the baseline data set (i.e., the images from the initial visits), we chose 58 severely diseased subjects and 58 healthy subjects according to their Kellgren-Lawrence score~\cite{Kellgren1957}. The sets were balanced in order to maximize the statistical power of our hypothesis test. For the 116 subjects, surfaces in correspondence of the distal femora and proximal tibiae were extracted from the respective 3D weDESS MR images (0.37×0.37 mm matrix, 0.7 mm slice thickness). The correspondence was then used to consistently extract (i.e., cut out) the condylar regions of the femura and tibiae; results from one subject can be seen on the left of \cref{fig:knee}.
        A supervised post-processing procedure was used to ensure the quality of the segmentations and the correspondence of the resulting meshes (8,988 vertices for each femur and 8,320 for each tibia).
        During the process, the relative positions of femur and tibia, which are present in the raw data through the global coordinate system of the scanner, are maintained. 
        
    \subsubsection{Encoding Relative Positions in SE(3)} \label{sec:encoding_SE3}
        
        \begin{figure*}[tb]
            \center       
            \includegraphics[width=0.90\linewidth]{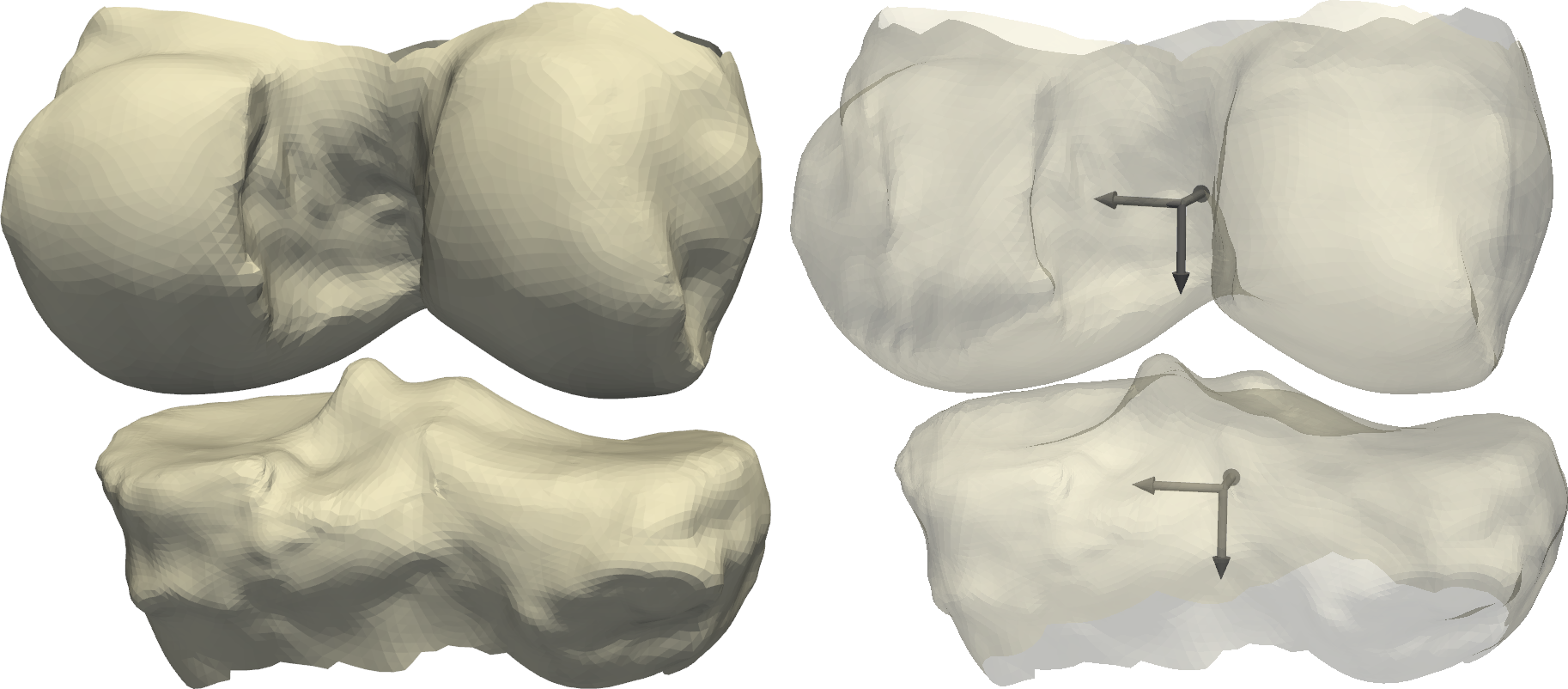}
            \caption{Visualization of the knee bones and the computed orthonormal frames. On the left, the distal femur (top) and the proximal tibia (below) are shown; the meshes (and their relative position) were reconstructed from OA data. The space between both (i.e., the joint space) is filled with cartilage, ligaments, and menisci (all not shown). On the right, we depict the orthonormal frames that correspond to both meshes.}
            \label{fig:knee} 
        \end{figure*}
    
        Remember that, with $O(n)$ denoting the group of orthogonal matrices in dimension $n$, the Euclidean group $\E(3) = \textnormal{O}(3) \ltimes \mathbb{R}^3$ is a semidirect product with group operation
        $(R, v)(Q, w) = (RQ, v+Rw)$ for $(R,v), (Q,w) \in \E(3)$. The group of rigid-body motions is the subgroup $\SE(3) = \SO(3) \ltimes \mathbb{R}^3$.
        
        Since for each femur-tibia pair the relative position is encoded in the coordinates of the meshes, we can capture it by calculating the rigid body transformation in SE(3) that moves the distal femur onto the proximal tibia. 
        More precisely, let $(O_f, o_f), (O_t, o_t) \in \E(3)$ be local orthonormal reference frames for the femur and tibia, respectively. (The vector part yields the origin, while the matrix columns represent the coordinate axes.) Then, after translating the meshes so that $o_{f}$ is moved to the origin (i.e., applying $(0, -o_{f}) \in \SE(3)$ to both frames),
        the rotation-translation pair $P = (R, v) \in \SE(3)$ that maps the reference frame $(O_f, o_f)$ of the femur to the reference frame $(O_t, o_t)$ of the tibia encodes the relative position of the bones. Now, since
        $$(O_t O_f^T, o_t - o_f)(O_f, 0) = (O_t, o_t - o_f),$$
        we have
        $$P = (R, v) = (O_t O_f^T, o_t - o_f).$$
        Now, in order to compute this we need to choose frames.
        We determine each reference frame from the principal component analysis of the vertices of the corresponding triangle mesh: While we pick the center of gravity as origin, the axes are chosen as unit vectors pointing along the principal directions; the latter is done in a consistent way throughout the population for both bone types; the resulting frames from one subject are visualized on the right of \cref{fig:knee}.
        
        Note that, in this context, left and right translations of $P$ with elements from SE(3) correspond to different choices of reference frames (with the same orientations) for the femur and tibia, respectively. 
        
        Finally, we want to point out that this approach of representing relative positions by elements of the special Euclidean group is also used in many other applications, for example, medical image analysis~\cite{Boisvert_ea2008}, robotics~\cite{ParkBobrowPloen1995}, human action recognition~\cite{vemulapalli2014human}, radar detection~\cite{Barbaresco2020,Cesic_ea2016}, and state description of molecules~\cite{BenoitHolmPutkaradze2011}. Thus, the permutation test outlined below can directly be applied to data from these domains.
        
    \subsubsection{Bi-invariant Permutation Test} \label{sec:permutation_test_single}
        We use a non-parametric permutation setup to test whether there is a difference between the knee configurations of people with OA and healthy controls. The reason is that we do not want to make assumptions about how both groups are distributed, but rather learn this from the data.
        After the 116 femur-tibia pairs have been processed as described in~\Cref{sec:encoding_SE3}, we obtain a set $(P_1^{(\textnormal{H})},\dots,P_{58}^{\textnormal{(H)}})$ in SE(3) of transformations derived from healthy controls and a set $(P_1^{\textnormal{(OA)}},\dots,P_{58}^{\textnormal{(OA)}})$ coming from patients with OA. We consider the null hypothesis $H_0$ of equal distributions underlying both groups, that is, $P^{\textnormal{(H)}} \stackrel{H_0}{\sim} P^{\textnormal{(OA)}}$.
        As test statistic $T$ we use (and thus compare) both the bi-invariant Hotelling $T^2$ statistic $t^2$ and the bi-invariant Bhattacharyya distance $\textnormal{D}_B$. 
        
        We initialize the test by computing the baseline 
        $$T_0 = T((P_i^{\textnormal{(H)}}), (P_i^{\textnormal{(OA)}})).$$
        Then, we perform 10,000 random permutations of the joint set
        $$(Z_1,\dots,Z_{116}) = \left( P_1^{\textnormal{(H)}},\dots,P_{58}^{\textnormal{(H)}}, P_1^{(\textnormal{OA})},\dots,P_{58}^{\textnormal{(OA)}} \right);$$
        that is, denoting the $l$-th permutation by $\sigma_l$, we compute the values
        $$T_l = T\big( (Z_{\sigma_l(1)},\dots,Z_{\sigma_l(58)}),(Z_{\sigma_l(59)},\dots,Z_{\sigma_l(116)}) \big), \quad l=1,\dots,10000.$$
        With $\boldsymbol{1}_{a\ge b}$ being $1$ if $a \ge b$ and 0 else,
        the $p$-value for the statistic $T$ is then given by 
        $$p_T= \frac{1}{10000} \sum_{l=1}^{10000} \boldsymbol{1}_{T_l \ge T_0};$$
        it is the proportion of test statistics that are greater than the one computed for the original (unpermuted) groups. A standard level to reject the null hypothesis is $p_T<0.05$; in this case, the difference between the distributions is called \emph{significant}.
        Performing the test with each statistic, we obtain 
        $$p_{t^2} \approx 0.00019 \quad \text{and} \quad p_{\textnormal{D}_B} < 10^{-5}.$$ 
        Thus, the differences are significant for both statistics. This result is in accordance with what is known about OA: Part of the disease is the so-called joint space narrowing---a decrease in distance between femur and tibia as consequence of degenerating cartilage. Note that we can deduce a posteriori that the sample size was sufficient to detect significant differences with an error probability smaller than $0.05$.
    
    \subsection{Hippocampi Shapes under Alzheimer's}
        In the second part, we analyze hippocampal atrophy patterns due to mild cognitive impairment (MCI) by applying and thereby comparing the derived Hotelling $T^2$ statistic and Bhattacharyya distance. MCI is a common condition in the elderly and often represents an intermediate stage between normal cognition and Alzheimer’s disease.
        As is consistently reported in neuroimaging studies, atrophy of the hippocampal formation is a characteristic early sign of MCI.
        Using both a local and global two-sample test, we infer significant differences in distribution of shapes of right hippocampi between a cognitive normal and the MCI group, in agreement with the literature.
    
    \subsubsection{Data Description}
    \label{sec:data}
    For our experiments we prepared a data set consisting of 26 subjects showing mild cognitive impairment (MCI) and 26 cognitive normal (CN) controls from the open access Alzheimer's Disease Neuroimaging Initiative\footnote{\url{http://adni.loni.usc.edu/}} (ADNI) database.
    ADNI provides, among others, 1632 brain magnetic resonance images collected on four different time points with segmented hippocampi.
    We established surface correspondence (2280 vertices, 4556 triangles) in a fully automatic manner employing the deblurring and denoising of functional maps approach~\cite{EzuzBenchen2017} for isosurfaces extracted from the available segmentations.
    The data set was randomly assembled from shapes for which segmentations were simply connected and remeshed surfaces were well-approximating ($\leq{}10^{-5}\;$mm root mean square surface distance to the isosurface).
    
\subsubsection{Lie Group-based Shape Space} \label{sec:shape_space}
    
    To model the shapes of the hippocampi, we use the representation from~\cite{AmbellanZachowvonTycowicz2019_GL3}. Given $n$ homogeneous objects in the form of triangle meshes $T_i \subset \mathbb{R}^3$ that are in correspondence and Procrustes aligned~\cite{Goodall1991}, their shapes are described in terms of (a power) of the Lie group $\Glp[3]$ of real 3-by-3 matrices with positive determinant.
    To obtain this representation, we view each mesh $T_i$ as the result of a deformation $\phi_i$ of a common reference $\overline{T} \subset \mathbb{R}^3$; that is, each $\phi_i: \overline{T} \to T_i$ is an orientation-preserving, simplicial isomorphism that yields a semantic correspondence. The underlying idea is that the derivatives of these deformations comprise the shape information, because they provide a local characterization of the deformation without being influenced by the location in space. Thus, the Jacobian $D\phi_i$ constitutes the shape representation of the $i$-th object. 
    Let $m$ be the (necessarily same) number of triangles of each mesh.
    Since each $D \phi_i$ is constant on the faces $F_j$ of $\overline{T}$, there are 3-by-3 matrices $G^{(i)}_j$ with $\textnormal{det}(G^{(i)}_j) >0$ such that 
    $D \phi_i \big|_{F_j} = G^{(i)}_j \in \Glp[3]$ for all $i=1,\dots,n,$ and  $j=1,\dots,m.$
    The $i$-th shape is thus given by
    $(G^{(i)}_1,\dots,G^{(i)}_m) \in \Glp[3]^m.$
    Note that the above $\textnormal{GL}^+(3)$ model stems from a continuous formulation~\cite[Sec.\ 2]{vonTycowicz_ea2018} that admits consistent and convergent discretizations.

    The task of obtaining a surface immersion whose Jacobian is closest to given differential coordinates leads to a variational problem.
    Its minimizer is given by the solution of the well-known Poisson equation for which fast numerical methods exist.
    Furthermore, as a global variational approach, the minimizer given by the Poisson equation tends to distribute errors uniformly such that local gradient field inconsistencies are attenuated.

    \subsubsection{Hippocampal Atrophy Patterns in CN vs.\ MCI}

    \begin{figure*}[tb]
        \center       
        \includegraphics[width=0.90\linewidth]{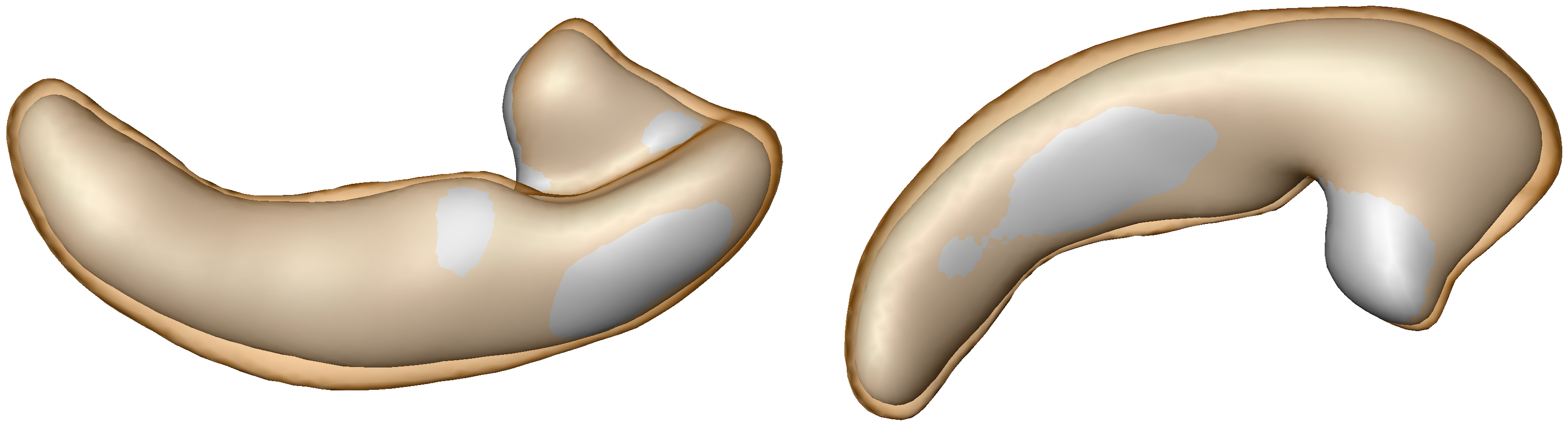}
        \caption{Group means of right hippocampi for cognitive normal (orange, transparent) and impaired (grey) subjects overlaid.}
        \label{fig:means} 
    \end{figure*}
    
    To test for a difference in distribution between the CN and MCI groups we extend the permutation test setup from~\Cref{sec:permutation_test_single} and perform a local (i.e., triangle-wise) and a global test. 
    After encoding the shapes as described in~\Cref{sec:shape_space}, the group means for the CN and MCI set are computed with the iterative algorithm from~\cite[Alg.\ 1]{PennecArsigny2013}. A well-known~\cite{Mueller_ea2010} phenomenon in patients that develop MCI is a loss of total hippocampal volume; it can be clearly observed in the qualitative comparison of both mean shapes shown in~\cref{fig:means}.
    
    \subsubsection{Local Test}
    In the following, when we speak of the $i$-th differential coordinate we mean the $i$-th element of the differential coordinates (vector). 
    For $i=1,\dots,4556$ let $G_i^{(\text{CN})}$ and $G_i^{\text{(MCI)}}$ denote the $i$-th differential coordinate of the mean hippocampi from the CN and MCI group, respectively. (We choose the first mesh of the CN group as reference object).
    In order to identify subregions that contribute to differences in mean shape between the groups, we perform triangle-wise, partial tests. For every triangle \textit{independently} we consider the null hypothesis $H_0$ that the distribution of its differential coordinate is the same for both groups, that is, 
    $G_i^{(\text{CN})} \stackrel{H_0}{\sim} G_i^{\text{(MCI)}}.$ 
    
    To test this hypothesis, we perform a permutation and compare again the bi-invariant Hotelling $T^2$ statistic and Bhattacharyya distance as test statistics.
    First, we perform independent tests for each triangle using the testing procedure from~\Cref{sec:permutation_test_single}. Therefore, denoting the test statistic by $T$ again, we start by computing 
    $$T^i_0 = T\big((G_i^{\text{(CN)}}), (G_i^{(\text{MCI})}) \big).$$ 
    We then perform 10,000 (random) permutations of the full set $$(Z_{i,1},\dots,Z_{i,52}) = \left( G_{i,1}^{\text{(CN)}},\dots,G_{i,26}^{\text{(CN)}},G_{i,1}^{\text{(MCI)}},\dots,G_{i,26}^{\text{(MCI)}} \right);$$
    here the second subscript enumerates the subjects of the groups. Then, we compute the statistics
    $$T^i_l = T\big( (Z_{i,\sigma_l(1)},\dots,Z_{i,\sigma_l(26)}),(Z_{i,\sigma_l(27)},\dots,Z_{i,\sigma_l(52)}) \big), \quad l=1,\dots,10000.$$
    The $p$-value of the $i$-th triangle for test statistic $T$ is then given by 
    $$p_T^{(i)} = \frac{1}{10000} \sum_{l=1}^{10000} \boldsymbol{1}_{T^i_l \ge T^i_0}.$$
    For each triangle, we reject $H_0$, if $p_T^{(i)} < 0.05$ and call the difference between both groups \textit{significant}.
    Using the bi-invariant Hotelling $T^2$ statistic and Bhattacharyya distance as test statistics, the described test is bi-invariant (in this case in~$\Glp[3]$). In particular, because of right-invariance, the results are independent from the reference that was chosen to compute the differential coordinates. (Left-invariance would allow to \textit{jointly} transform the differential coordinates of the subjects by elements from $\Glp[3]$, which, e.g., could be of interest for a better numerical performance of an algorithm.)
    
    Because of the large number of tests we apply Benjamini-Hochberg false discovery correction at the level $\alpha = 0.05$ to identify triangles with significant differences, for which we reject $H_0$.
    In~\cref{fig:p-values}, we visualize the triangles with significant differences, showing the respective $p$-values for both statistics. 
    
    In line with literature on MCI~\cite{Mueller_ea2010}, the obtained results suggest more differentiated morphometric changes beyond homogeneous volumetric decline of the hippocampi. The Bhattacharyya distance detects more significant differences, which is expected because it also takes covariances into account. Interestingly, there are also some areas where the Hotelling $T^2$ statistic seems to be more sensitive. This is probably the case when, for the Bhattacharyya distance, small but significant differences in the mean become insignificant due to largely coinciding covariances.
        
    \begin{figure*}[tb]
        \center       
        \includegraphics[width=0.90\linewidth]{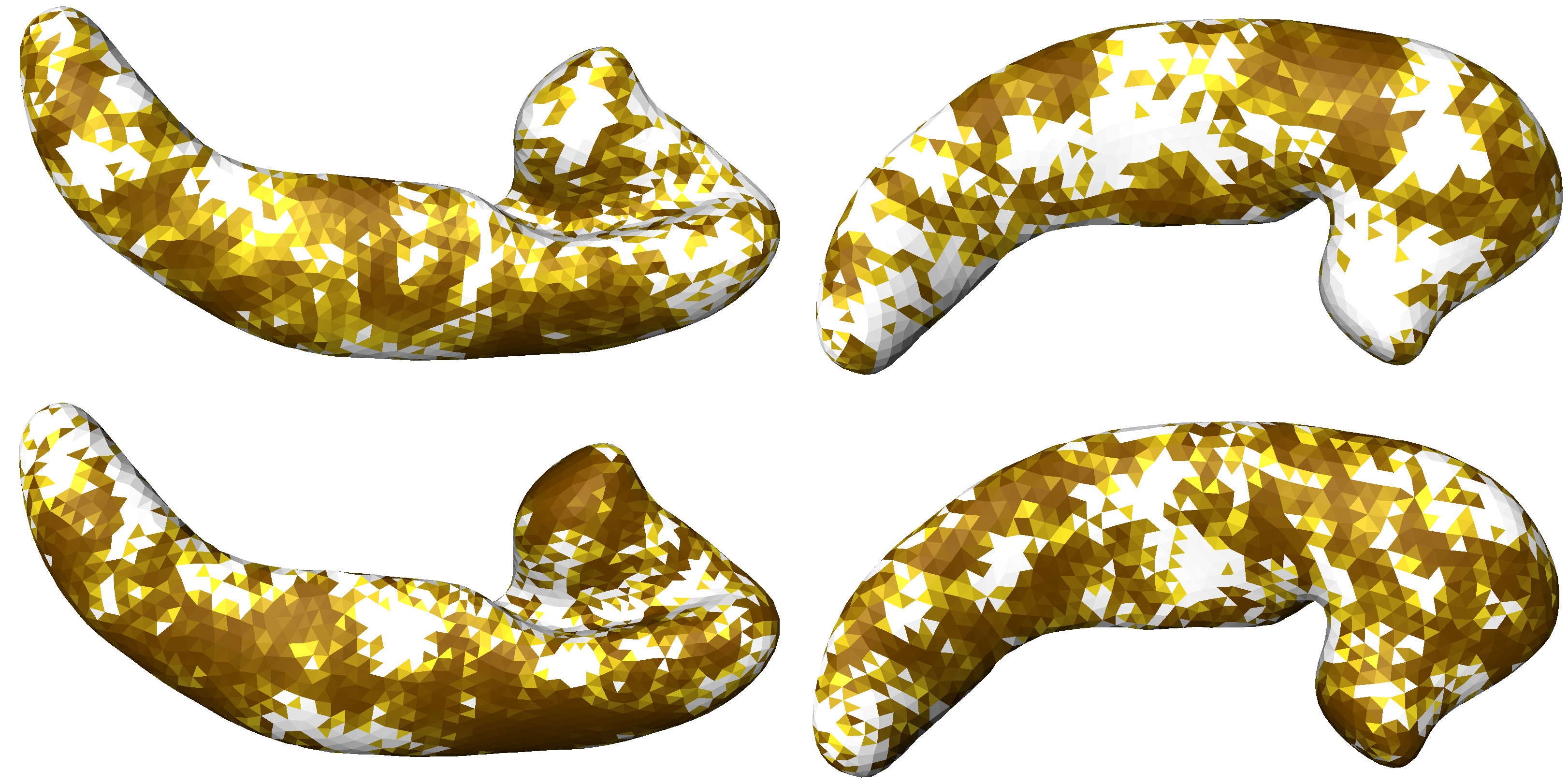}
        \caption{Group tests for differences in the distribution of right hippocampi for cognitive normal and impaired subjects. Results for the bi-invariant Hotelling $T^2$ statistic are shown at the top and for the Bhattacharyya distance at the bottom. Each triangle of the CN mean is color coded according to its $p$-value (FDR corrected) using the colormap 0.0 \ShowColormap{} 0.05.}
        \label{fig:p-values} 
    \end{figure*}
    
    \subsubsection{Global Test}
    Complementary to the local test is a global test that includes the overall pattern of the covariance structures thus being sensitive to spatial dependencies across the shape.
    In the construction of the test we follow the approach from~\cite{Schulz_ea2016}.
    
    For each triangle we first map the $T^i_l$ to an approximate uniform distribution in $[0,1]$ by applying the corresponding empirical cumulative distribution function (cdf) $C_i$. More precisely, we compute
    $$ C_i(T^i_l) = \frac{1}{10000} \sum_{r=1}^{10000} \boldsymbol{1}_{T^i_r \le T^i_l}, \quad i=1,\dots,4556, \quad l=0,\dots,10000.$$
    Then, setting $\widetilde{C}_i(T_l^i) = 0.9998\, C_i(T_l^i) - 0.00001$ and denoting the cdf of the standard normal distribution by $\phi$, the data is mapped to $U_l^i = \phi^{-1}(\widetilde{C}_i(T^i_l))$. The latter follow an approximate standard normal distribution for each triangle.

    Using the sample covariance matrix $\Sigma = 1/9999\,UU^T$ (where $U = [U_1,\dots,U_{10000}] = [U^i_l]$), a suitable test statistic is now given by the squared Mahalanobis distance; this yields
    \begin{equation} \label{eq:global_Mahalanobis}
    M_0 = U_0^T\Sigma^{-1} U_0,\quad M_l = U_l^T\Sigma^{-1} U_l, \quad l=1,\dots,10000.
    \end{equation}
    We then call the difference in global shape significant, when 
    $$p = \frac{1}{10000} \sum_{l=1}^{10000} \boldsymbol{1}_{M_l \ge M_0} < 0.05.$$
    In order to account for the irregularity in surface triangulations, the sample covariance operator and thus Mahalanobis distance in~\cref{eq:global_Mahalanobis} can be extended in terms of an adapted inner product that weights each component by the corresponding triangle's area.
    In this experiment, we adhere to the standard formulation~\cref{eq:global_Mahalanobis}, as the studied hippocampi surfaces are meshed uniformly.
    
    Applying this test to the values obtained from the local test yielded $p$-values smaller than $1/10^5$ for the bi-invariant Hotelling $T^2$ statistic and Bhattacharyya distance revealing that both are sensitive to the underlying morphological changes. 
    
\section{Discussion}
    In this work, we generalized established dissimilarity measures between sample distributions to Lie groups.
    These new indices are, in particular, invariant under translations, which yields analyses that are unbiased from the position of the data sets in the group. 
    As an application, we constructed nonparametric two-sample tests based on the proposed measures for both local and global hypothesis tests of shapes and used them for group tests on malformations of right hippocampi due to mild cognitive impairment.

    An interesting question, both theoretically and practically, is: Given a Lie group, to what extent is there smeariness of the group mean~\cite{EltznerHuckemann2019}? 
    An answer would give an indication on the size a data set should have for reliable results.
    
    In order for the Mahalanobis distances to be well-defined, the sample covariance operator needs to be invertible; this is frequently violated, most importantly when the number of observations is lower than the number of variables.
    A common approach in such situations is to resort to the pseudo-inverse of the covariance. Doing this in~\cref{eq:mahalanobis_distance}, however, will not result in a bi-invariant notion of Mahalanobis distance; the reason for this lies in the fact that taking pseudo-inverses of matrix products is in general more difficult than inverses.
    Since high dimensional data is common today, extending the proposed expressions to such scenarios poses another interesting direction for future work.

\appendix

\section{Existence and Uniqueness of the Group Mean} \label{app:existence_uniqueness}
    Let $G$ be a Lie Group with affine connection. If $U \subseteq G$ is a normal convex neighborhood and $g_1,\dots,g_m \in U$, then there always exists a group mean  $\overline{g} \in U$ \cite[Thm.\ 5.3]{PennecLorenzi2020}. 
    Nevertheless, it need not be unique. For the latter to be ensured we need a stronger notion of convexity that was proposed by Arnaudon and Li in \cite{ArnaudonLi2005}. In the following, we summarize their ideas; for more details, including examples, we refer to their article or \cite{PennecLorenzi2020} and the references therein.
    
    First, we need the following definitions. (Remember that a real-valued function on a smooth manifold is called convex, if its restriction to any geodesic is convex.)
    \begin{definition}
        A function $\rho: U \times U \to \mathbb{R}_{\ge 0}$ that is convex with respect to the product structure is called 
        \textnormal{separating function} if it vanishes on the diagonal of $U \times U$ and only there.
    \end{definition}
    \begin{definition}[$p$-convexity]
        Let $\rho$ and $d$ be a smooth separating function and an auxiliary Riemannian distance function on $U$, respectively. We say that $U$ has a \textnormal{$p$-convex geometry} if there are constants $c,C \in \mathbb{R}$ with $0 < c < C$ and an even integer $p \ge 2$ such that
        $$c d(f,g)^p \le \rho(f,g) \le C d(f,g)^p$$
        for all $f,g \in U$.
    \end{definition}
    For general smooth manifolds it is known that not all normal convex neighborhoods have a $p$-convex geometry. On the other hand, Whiteheads theorem ensures that each point in $G$ has a 2-convex neighborhood. Importantly, if the normal convex neighborhood $U$ has a $p$-convex geometry for any $p \in 2 \mathbb{N}$, then this is enough to ensure the uniqueness of the group mean.
    Nevertheless, there is a weaker condition that still yields uniqueness.
    \begin{definition}(CSLCG neighborhood~\cite{ArnaudonLi2005})
        We say that $U$ is \textnormal{convex with semilocal convex geometry (CSLCG)} if every compact subset $K \subset U$ has a relatively compact neighborhood $U_K$ with $p_K$-convex geometry for some $p_K \in 2\mathbb{N}$ depending on $K$.
    \end{definition}
    Observe that if $U$ has a $p$-convex geometry for some $p\in 2 \mathbb{N}$, then it is CSLCG. 
    We have the following result, which is a special case of \cite[Prop.\ 2.4]{ArnaudonLi2005}.
    \begin{proposition} \label{prop:existence_uniqueness_mean}
        Let $G$ be a Lie group with affine connection. Further, let $U \subseteq G$ be a CSLCG neighborhood and $g_1,\dots,g_m \in U$. Then, there exists a unique group mean $\overline{g} \in U$ of $g_1,\dots,g_m$.
    \end{proposition}
    
    The following lemma provides a property of CSLCG neighborhoods that is needed in this work.
    \begin{lemma} \label{lem:cslcg_diff_Rd}
        Let $G$ be a $d$-dimensional Lie group with affine connection and $U \subseteq G$ be a CSLCG neighborhood. Then, $U$ is diffeomorphic to $\mathbb{R}^d$.
    \end{lemma}
    \begin{proof}
        Let $g \in U$. Since $U$ is a normal convex neighborhood, $\Log_g(U) \subseteq T_gG$ is well-defined and star-shaped about $0 \in T_gG$. The claim now follows since any star-shaped domain in a $d$-dimensional vector space is diffeomorphic to $\mathbb{R}^d$~\cite[Thm.\ 5.1]{BottTu1982}.
    \end{proof}

\section{Hotelling's \texorpdfstring{{\boldmath$T^2$}}{T squared} statistic for Riemannian manifolds} \label{app:riemannian}
    In \cite[Sec.\ 3.3]{MuralidharanFletcher2012}, Muralidharan and Fletcher introduce a generalization of Hotelling's $T^2$ statistic to Riemannian manifolds $M$, i.e., for samples $(p_1,\dots,p_m)$, $(q_1,\dots,q_n)$ in $M$. The centers of the data sets are then given by the Fréchet means $\overline{p}, \overline{q} \in M$, respectively. Assuming that $\overline{p}, \overline{q}$ are unique, the difference between the means can be replaced by the Riemannian logarithms $v_{\overline{p}} = \Log_{\overline{p}} (\overline{q}) \in T_{\overline{p}}M$ or $v_{\overline{q}} = \Log_{\overline{q}} (\overline{p}) \in T_{\overline{q}}M$. Depending on the choice, the vectors are from different tangent spaces.
    Then, sample covariance matrices can be defined by
    \begin{align*}
        \textstyle
        \big[ W_{p_i} \big] = \frac{1}{m} \sum_{i=1}^m \big[\Log_{\overline{p}}(p_i) \big] \big[\Log_{\overline{p}}(p_i) \big]^T, \quad \big[W_{q_i} \big] = \frac{1}{n} \sum_{i=1}^n \big[\Log_{\overline{q}}(q_i) \big] \big[\Log_{\overline{q}}(q_i) \big]^T.
    \end{align*}
    Since there is no canonical way to compare vectors from different tangent spaces, Muralidharan and Fletcher propose to calculate a generalized $T^2$ statistic at both means and average the results. This leads to the generalized Hotelling $T^2$ statistic
    \begin{equation} \label{eq:Riemannian_test}
        \textstyle
        \mathsf{t}^2 \big((p_i),(q_i) \big) = \frac{1}{2} \left(\big[\Log_{\overline{p}}(\overline{q}) \big]^T \big[W_{p_i}^{-1} \big] \big[\Log_{\overline{p}}(\overline{q}) \big] + \big[\Log_{\overline{q}}(\overline{p}) \big]^T \big[W_{q_i}^{-1} \big] \big[\Log_{\overline{q}}(\overline{p}) \big] \right)
    \end{equation}
    for Riemannian manifolds.
    
    Note, however, that in the Euclidean case the statistics \cref{eq:Euclidean_test,eq:Riemannian_test} do not coincide. To see this, observe first that in this case $\Log_{\overline{p}}(\overline{q}) = - \Log_{\overline{q}}(\overline{p}) = \overline{q}- \overline{p}$ by the standard identification of tangent spaces. Hence, we get from \cref{eq:Riemannian_test} that
    \begin{equation*}
        \textstyle
        \mathsf{t}^2 \big((p_i),(q_i) \big) = (\overline{p} - \overline{q})^T  \frac{1}{2} \left(W_{p_i}^{-1} + W_{q_i}^{-1} \right) (\overline{p} - \overline{q}).
    \end{equation*}
    Thus, it is enough to show that $\widehat{S}^{-1} \ne 1/2 \left(W_{p_i}^{-1} + W_{q_i}^{-1} \right)$ in general for the pooled sample covariance $\widehat{S}$. But this is true---not only because of scaling---since for symmetric positive definite matrices $P_1, P_2$ (of the same size) we have $P_1^{-1} + P_2^{-1} \ne (P_1 + P_2)^{-1}$ in general (and any such matrix can be a pooled sample covariance as can be seen with the help of the singular value decomposition).

\section*{Acknowledgments}
    
    We are grateful for the open-access data sets of the Osteoarthritis Initiative (OAI)\footnote{Osteoarthritis Initiative is a public-private partnership comprised of five contracts (N01-AR-2-2258; N01-AR-2-2259; N01-AR-2-2260; N01-AR-2-2261; N01-AR-2-2262) funded by the National Institutes of Health, a branch of the Department of Health and Human Services, and conducted by the OAI Study Investigators. Private funding partners include Merck Research Laboratories; Novartis Pharmaceuticals Corporation, GlaxoSmithKline; and Pfizer, Inc. Private sector funding for the OAI is managed by the Foundation for the National Institutes of Health. This manuscript was prepared using an OAI public use data set and does not necessarily reflect the opinions or views of the OAI investigators, the NIH, or the private funding partners.} and the Alzheimer's Disease Neuroimaging Initiative (ADNI)\footnote{Data collection and sharing for this project was funded by the ADNI (National Institutes of Health Grant U01 AG024904) and DOD ADNI (Department of Defense award number W81XWH-12-2-0012). ADNI is funded by the National Institute on Aging, the National Institute of Biomedical Imaging and Bioengineering, and through generous contributions from the following: AbbVie, Alzheimer's Association; Alzheimer's Drug Discovery Foundation; Araclon Biotech; BioClinica, Inc.; Biogen; Bristol-Myers Squibb Company; CereSpir, Inc.; Cogstate; Eisai Inc.; Elan Pharmaceuticals, Inc.; Eli Lilly and Company; EuroImmun; F. Hoffmann-La Roche Ltd and its affiliated company Genentech, Inc.; Fujirebio; GE Healthcare; IXICO Ltd.; Janssen Alzheimer Immunotherapy Research \& Development, LLC.; Johnson \& Johnson Pharmaceutical Research \& Development LLC.; Lumosity; Lundbeck; Merck \& Co., Inc.; Meso Scale Diagnostics, LLC.; NeuroRx Research; Neurotrack Technologies; Novartis Pharmaceuticals Corporation; Pfizer Inc.; Piramal Imaging; Servier; Takeda Pharmaceutical Company; and Transition Therapeutics. The Canadian Institutes of Health Research is providing funds to support ADNI clinical sites in Canada. Private sector contributions are facilitated by the Foundation for the National Institutes of Health (www.fnih.org). The grantee organization is the Northern California Institute for Research and Education, and the study is coordinated by the Alzheimer's Therapeutic Research Institute at the University of Southern California. ADNI data are disseminated by the Laboratory for Neuro Imaging at the University of Southern California.}. 
    Furthermore, we are thankfull for F. Ambellan's help in establishing dense correspondences of the hippocampal surface meshes.

\bibliographystyle{siamplain}
\bibliography{bibliography}

\end{document}